\documentclass[11pt]{article}

\usepackage{amsfonts}
\usepackage{graphicx,subfig}
\usepackage{amsmath,amsthm}
\usepackage{wrapfig}
\usepackage{amssymb}

\usepackage{epsfig}
\usepackage{amsmath}
\usepackage{amssymb}
\usepackage{psfrag}
\usepackage[absolute]{textpos}
\usepackage{setspace}

\newtheorem{thm}{Theorem}[section]
\newtheorem{lem}{Lemma} [section]
\newtheorem{definition}{Definition}

\newtheorem{alg}{Algorithm}

\def\x{{\bf x}}
\def\y{{\bf y}}
\def\w{{\bf w}}

\def\C{{\mathcal{C}}}
\def\K{{\mathcal{K}}}
\def\H{H}

\def\G{{\mathcal{G}}}

\newcommand{\beq}{\begin{equation}}
\newcommand{\eeq}{\end{equation}}
\newcommand{\bea}{\begin{eqnarray}}
\newcommand{\eea}{\end{eqnarray}}

\newcommand{\stexp}{\mbox{$\mathbb{E}$}}    % Statistical Expectation operator
\newcommand{\Prob}{\ensuremath{\mathbb{P}}}
\long\def\symbolfootnote[#1]#2{\begingroup%
\def\thefootnote{\fnsymbol{footnote}}\footnote[#1]{#2}\endgroup}

\def\x{{\bf x}}
\def\y{{\bf y}}
\def\w{{\bf w}}

\newcommand{\karis}{{(\eps,d,r_0,n)}}

\newcommand{\eps}{\epsilon}

\newcommand{\h}{{\bf h}}
\newcommand{\z}{{\bf z}}

\begin{document}
\title{Reweighted LP Decoding for LDPC Codes}

\author{ Amin Khajehnejad$^*$ \, Alexandros G. Dimakis$^\dag$\, Babak Hassibi$^*$\\
  Benjamin Vigoda$^+$\, William Bradley$^+$\\ $^*$California Institute of Technology \\ $^\dag$University of Southern California \\$^+$Lyric Semiconductors Inc.
\thanks{This work was supported in part by the National Science Foundation under grants CCF-0729203, CNS-0932428 and CCF-1018927, by the Office of Naval Research under the MURI grant N00014-08-1-0747, by Caltech's Lee Center for Advanced Networking and by DARPA FA8750-07-C-0231.}
}

%\author{\IEEEauthorblockN{M. Amin Khajehnejad}
%\IEEEauthorblockA{Department of Electrical Engineering\\
%California Institute of Technology\\
%Pasadena, California 91125\\
%Email: amin@caltech.edu}
%\and
%\IEEEauthorblockN{Weiyu Xu}
%\IEEEauthorblockA{Department of Electrical Engineering\\

\maketitle

\begin{abstract}
We introduce a novel algorithm for decoding binary linear codes
by linear programming. We build on the LP decoding algorithm of
Feldman \textit{et al.}
and introduce a post-processing step that solves a second linear program
that reweights the objective function based on the outcome of the
original LP decoder output.
Our analysis shows that for some LDPC ensembles we can improve the
provable threshold guarantees
compared to standard LP decoding. We also show significant empirical
performance gains for the reweighted LP decoding algorithm with very
small additional computational complexity.
\end{abstract}

\section{Introduction}
\label{sec:Intro}

Linear programming (LP) decoding for binary linear codes was
introduced by
Feldman, Karger and Wainwright~\cite{Feldman Wainwright}.
The method is based on solving a linear-programming relaxation
of the integer program corresponding to the maximum likelihood (ML)
decoding problem.
LP decoding is connected to message-passing decoding~
\cite{FelKarWai02,Wainwright02aller},
and graph covers~\cite{KoeVon03,VonKoe06b} and has received
substantial recent attention
(see e.g. \cite{VonKoe06b}, and \cite{TagSie06}). %\textbf{BBBDaskalakis
%et al and other papers here}).

As with the work described here, a related line of work has studied
various improvements to either standard iterative
decoding~\cite{Fossorier01,PishroFekri} or to LP decoding via
nonlinear extensions~\cite{YanFelWan06} or loop
corrections~\cite{CheChe06}.

The practical performance of LP decoding is roughly comparable to min-
sum decoding
and slightly inferior to sum-product decoding.
In contrast to message-passing decoding, however, the LP decoder
either concedes
failure on a problem, or returns a codeword along with a guarantee
that it is the ML codeword, thereby eliminating any undetected
decoding errors.

The main idea of this paper is to add a second LP as a post-processing
step when
original LP decoding fails and outputs a fractional pseudocodeword. We
use the difference
between the input channel likelihood and the pseudocodeword coordinate
to find a measure of disagreement or unreliability for each bit. We
subsequently use this unreliability to
bias the objective function and re-run the LP with the reweighted
objective function.  The reweighting increases the cost of changing
reliable bits and decreases the cost for unreliable bits.
We present an analysis that the provable BSC recovery
thresholds improve for certain families of LDPC codes. We stress that
the actual thresholds, even for the original LP decoding algorithm,
remain unknown.
Our analysis only establishes that the obtainable lower bounds on the
fraction of recoverable errors
are improved compared to the corresponding bounds for LP decoding. It
is possible, however, that
this is just an artifact of the lower bound techniques and that the
true threshold is identical for both algorithms. In any case, the
empirical performance gains we observe in our preliminary experimental
analysis seem quite substantial.

A central idea in our analysis is a notion of robustness
to changes in the BSC bit-flipping probability. This concept was
inspired by a similar reweighted iterative $\ell_1$ minimization
idea for compressive sensing~\cite{CWB07,Iterative l1 ISIT}.
We note that the reweighting idea of this paper involves changing the
objective function of the LP from the reweighted max-
product
algorithm~\cite{WaiJaaWil05b}.

\section{Basic Definitions}
\label{sec:Definitions} A vector $\x$ in $\mathbb{R}^n$ is called
$k$-sparse if it has exactly $k$ nonzero entries. The support set of
a sparse vector $\x$ is the index set of its nonzero entries. If
$\x$ is not sparse, the $k$-support set of $\x$ is defined as the
index set of the maximum $k$ entries of $\x$ in magnitude.  We use
$\|\x\|_p$ to denote the $\ell_p$ norm of a vector $\x$ for $p\geq 0$.
in particular $\|\x\|_0$ is defined to be the number of nonzero
entries in $\x$. For a set $S$, cardinality of $S$ is denoted by
$|S|$ and if $S\subset\{1,2,\cdots,n\}$, then $\x_S$ is the
sub-vector formed by those entries of $\x$ indexed in $S$. Also the
complement set of $S$ is denoted by $S^c$. The rate of a linear binary code
$\C$ is denoted by $R$, and the corresponding parity check matrix is
$H\in \mathbb{F}^{m\times n}$, where $n$ is the length of each
codeword and $m=Rn$. The factor graph corresponding to $\C$ is denoted
by $\mathcal{G} = (X_v,X_c,\mathcal{E})$, where $X_v$ and $X_c$ are
the sets of variable nodes and check nodes respectively, and
$\mathcal{E}$ is the set of edges. For regular graphs, $d_v$ and $d_c$ denote the degree of variable and check nodes
respectively. The girth of a graph $\mathcal{G}$, denoted by $\text{girth}(\mathcal{G})$, is defined to be
the size of the smallest cycle in $\mathcal{G}$.

\section{Background}
\label{sec:background}

Suppose that $C$ is a memoryless channel with binary input and an
output alphabet $\mathcal{Y}$, defined by the
transition probabilities $P_{Y|X}(y|x)$. For a received symbol $y$,
the likelihood ratio is defined as
$\log(\frac{P_{Y|X}(y|x=0)}{P_{Y|X}(y|x=1)})$, where $x$ is the
transmitted symbol. If a codeword $\x^{(c)}$ of length $n$ from
the linear code $\C$ is transmitted through the channel, and an output
vector $\x^{(r)}$ is received, a maximum likelihood decoder can be
used to estimate the transmitted codeword by finding the \emph{most
likely} transmitted input codeword. Let $\gamma_i$ be the likelihood
ratio assigned to the $i^\text{th}$ received bit $\x^{(r)}_i$, and $\gamma$
be the likelihood vector  $\gamma = (\gamma_1,\cdots,\gamma_n)^T$.
The ML decoder can be formalized as follows~\cite{feldthesis}

\bea \nonumber \text{ML decoder:}&&\text{minimize} ~\gamma^T \x \\
&&\text{subject
    to}~\x\in\text{conv}(\C), \label{ML decoder}\eea

\noindent where $\text{conv}(\C)$ is the convex hull of all the
codewords of $\C$ in $\mathbb{R}^n$. The linear program (\ref{ML
  decoder}) solves the ML decoding problem by the virtue of the fact that
the objective $\gamma^T\x$ is minimized by a corner point (or vertex)
of $\text{conv}(\C)$, which is necessarily a codeword (In fact, vertices of
$\text{conv}(\C)$ are all the codewords of $\C$). In a linear program,
the polytope over which the optimization is performed is described by
linear inequalities describing the facets of the polytope. Since decoding for general linear codes is NP hard, it is unlikely that $\text{Conv}(\C)$ can be efficiently described. Feldman \textit{et al.} introduced a relaxation of (\ref{ML decoder}) by  replacing the polytope $\text{conv}(\C)$ with a new
polytope $\mathcal{P}$ that has much fewer facets, contains $\text{conv}(\C)$ and retains the codewords
of $\C$ as its vertices~\cite{feldthesis}. One way to construct $\mathcal{P}$ is
the following. If the parity check matrix of $\C$ is the $m\times n$
matrix $H$ and if $\h_j^T$ is the $j$-th row of $H$, then

\beq \mathcal{P} = \cap_{1\leq j\leq m} \text{conv}(\C_j), \eeq

\noindent where $\C_j = \{\x\in \mathbb{F}^n~|~ \h_j^T\x =  0
~\text{mod}~ 2 \}$.  As mentioned earlier, with this construction, all
codewords of $\C$ are vertices of
$\mathcal{P}$. However, $\mathcal{P}$ has some additional
vertices with fractional entries in $[0,1]^n$. A vertex of the polytope
$\mathcal{P}$ is called a \emph{pseudo-codeword}. Moreover, if a
pseudo-codeword is integral, i.e., if it has 0 or 1 entries, then it
is definitely a codeword. The LP
relaxation of (\ref{ML decoder}) can thus be written as:

\bea \nonumber \text{LP decoder:}&&\text{minimize} ~\gamma^T \x \\
&&\text{subject
    to}~\x\in \mathcal{P}. \label{LP decoder}\eea

The number of facets of $\mathcal{P}$ is exponential in the maximum
weight of a row of $H$. Therefore, for LDPC codes with a small (often constant) row density, $\mathcal{P}$ has a
polynomial number of facets, and it is possible to solve (\ref{LP decoder}) in polynomial time.

For binary symmetric channels, (\ref{LP decoder}) has another useful
interpretation. In this case, rather than minimize $\gamma^T\x$ it
turns out that one can alternatively minimize the Hamming distance
between the output of the channel $\x^{(r)}$ and the individual
codewords $\x\in\C$. Using the fact that the LP relaxation with
$\mathcal{P}$ relaxes the entries of $\x$ from $x_i\in\{0,1\}$ to
$x_i\in [0,1]$, we may replace the Hamming distance with the $\ell_1$
distance $\|\x-\x^{(r)}\|_1$. This implies that the decoder (\ref{LP
  decoder}) is equivalent to

\bea \nonumber \text{BSC-LP decoder:}&&\text{minimize}~ \|\x-\x^{(r)}\|_1 \\
&&\text{subject
    to}~\x\in\mathcal{P}. \label{BSC-LP decoder}\eea
\noindent The above formulation can be interpreted as follows. For a
received output binary vector $\x^{(r)}$, the solution to the LP
decoder is basically the closest (in the $\ell_1$ distance sense)
pseudo-codeword to $\x^{(r)}$.

Linear programming decoding was first introduced by Feldman \textit{et al.} \cite{feldthesis,Feldman Wainwright}. Subsequently \cite{Feldman constant fraction of error} it was shown that if the parity check matrix
is chosen to be the adjacency matrix of a high-quality expander, LP
decoding can correct a constant fraction of errors. A fundamental lemma in \cite{Feldman Wainwright}
and used in the results therein, is that the LP polytope
$\mathcal{P}$ is the same polytope from the view point of every
codeword, and therefore for the analysis of LP decoding, it can be
assumed without loss of generality that the transmitted codeword is
the all zero codeword. The theoretical results of \cite{Feldman
constant fraction of error} were based on a dual witness argument,
i.e. a feasible set of variables that set the dual of LP equal to
zero. However, the bounds on success threshold of LP decoding achieved by this technique is considerably smaller than the empirical recovery threshold of LP decoder in practice. A later analysis of LP decoding by Daskalakis \textit{et al.} \cite{Alex Decoding} improved upon those bounds for random expander
codes, through employing a different dual witness argument, and
considering a \emph{weak} notion of LP success rather than the
\emph{strong} notion of \cite{Feldman constant fraction of error}. A
strong threshold means that \emph{every} set of errors of up to a
certain size can be corrected, whereas a weak threshold implies that
\emph{almost all} error sets of a certain size are recoverable.
Note that there is a gap of about one order of magnitude between the error-correcting thresholds of \cite{Alex Decoding}  and the ones observed in practice.

The arguments of \cite{Feldman constant fraction of error} and
\cite{Alex Decoding} are based on the existence of dual certificates that guarantee the
success of the LP decoder and require codes that are based
on bipartite expander graphs. A more recent work of Arora \textit{et
al.}~uses a quite different certificate based on the primal LP
problem \cite{ADS}. This approach results in fairly easier
computations and significantly better thresholds for LP decoding.
However, the underlying codes discussed in \cite{ADS} are based on factor graphs
with a large girth (at least doubly logarithmic in the number of variables), rather than unbalanced expanders considered in previous arguments. Note that similar to \cite{Alex Decoding}, the bounds of
\cite{ADS} are weak bounds, certifying that for a random set of errors up to a fraction of bits, LP decoding succeeds with high probability. The largest such fraction is called the weak recovery threshold.

A somewhat related problem to the LP decoding of linear codes is
the compressed sensing (CS) problem. In CS an
unknown real vector $\x$ of size $n$ is to be recovered from a set of
$m$ linear measurement, represented by $\y = A\x$, where $A\in \mathbb{R}^{m\times n}$, and $m << n$. This is in general infeasible, since the measurement matrix $A$ is under-determined and
the resulting system of equations is ill-posed, i.e., it can have infinitely many solutions. However,
imposing a sparsity condition on $\x$ can make the solution
unique. The unique sparse solution can be found by exhaustive search
for instance, which is formulated by the following minimization
program: \bea \nonumber &&\text{minimize}~ \|\x\|_0 \\ \label{eq:l0
min}&&\text{subject to}~ A\x = \y. \eea

\noindent Since (\ref{eq:l0 min}) is NP-hard, one possible
approximation is relaxing the $\ell_0$ norm of $x$ to the
closet convex norm $\|\x\|_1$, which results in the following $\ell_1$ minimization program:

\bea  &&\text{minimize}~ \|\x\|_1 \\ \label{eq:l1 min}&&\text{subject
to}~ A\x = \y.\eea

\noindent (\ref{eq:l1 min}) is a linear program, which can in general be solved in polynomial time. There has been substantial theoretical work on this linear programming relaxation, see \textit{e.g.} \cite{Donoho,Donoho1,Null Space1,Null Space2,Weiyu robustness}

Recently, systematic connections between the problems of channel coding LP and CS $\ell_1$ relaxation has been found~\cite{Alex Vontobel,Alex_Vontobel2}. In this paper, we build on those connections to improve LP decoding, and further extend the ideas of robustness and reweighted $\ell_1$ minimization in compressed sensing to channel coding LP.

\section{Extended Certificate and Robustness of LP decoder}
\label{sec:robutsness} The success of LP decoder is often certified
by the existence of a \emph{dual witness}  \cite{Feldman constant
fraction of error,Alex Decoding}. Similarly, for $\ell_1$
minimization in the context of CS, a dual witness
certificate can guarantee that the recovery of sparse signals is
successful \cite{Candes RIP}.  However, it has proven more promising to express the success condition of $\ell_1$ minimization in
terms of  the properties of the null space of the measurement
matrix~\cite{Null Space1,Null Space2,Null Space3}. The condition is
called \emph{null space property}, through which it is possible to characterize one class of ``good'' measurement matrices for CS, namely matrices that are congruent with $\ell_1$ minimization decoding. The advantage of the null space interpretation, apart from the fact that it results in sharper analytical bounds, is that with proper
parametrization, it can also be used to evaluate the performance of
$\ell_1$ minimization in the presence of noise. This is known as
the \emph{robustness} of $\ell_1$ minimization.  A consequence of
the robustness property is that when $\ell_1$ minimization fails to
recover a sparse signal, it often gives a decent approximation to it
\cite{Iterative l1 ISIT}. To the best of our knowledge, a similar
certificate has not been introduced in the context of channel coding
linear programming. In other words, when LP decoding fails to return an integral solution, it is not known how far in the proximity of the actual codeword it lies. We provide an approximate solution to this question in this section, using the following strategy. We introduce a property called fundamental cone property for an
arbitrary code $\C$, and show that for binary symmetric channels,
this is related to the robustness of the solution of the LP decoder.
The robustness of LP decoding has two consequences. First, it
implies that the linear program is tolerant to a limited mismatch in the available formulation. Second, it can be used to develop iterative schemes that improve the performance of the decoder.  We
will discuss these issues in proceeding sections.  We begin by defining the
fundamental cone of a code from \cite{Alex Vontobel}.

\begin{definition}
Let $H$ be a parity check matrix. Define $\mathcal{J}$ and
$\mathcal{I}$ to be the set of rows and columns of $H$. Also, for
each $j\in \mathcal{J}$, define $\mathcal{I}_j = \{i \in
\mathcal{I}~|~ H(j,i) = 0\}$. The fundamental cone, $\K(H)$, of $H$ is the set of all vectors $\omega=(\omega_1,\omega_2,\dots,\omega_n)^T$ that satisfy

\begin{align} & \omega_i\geq 0, ~~~~~~~~~~~~ \forall 1\leq i\leq n, \\
& \omega_i \leq \sum_{i'\in \mathcal{I}_j \setminus i}\omega_{i'},
~~\forall j\in \mathcal{J}~\forall i\in \mathcal{I}_j. \end{align}
\end{definition}

\noindent $\K(\H)$ is the smallest cone in $\mathbb{R}^n$ that
encompasses the polytope $\mathcal{P}$. If a vector lies on an edge
of $\K$, it is called a \emph{minimal pseudo-codeword}. For simplicity, in the sequel, we use $\mathcal{K}$ instead of $\mathcal{K}(H)$ whenever there is no ambiguity.

\begin{definition}
\label{def:FCP} Let $S \subset \{1,2,\cdots,n\}$ and $C\geq1$ be
fixed. A code $\mathcal{C}$ with parity check matrix $\H$ is said to
have the fundamental cone property $\text{FCP}(S,C)$, if for every nonzero vector $\omega
\in \mathcal{K}(\H)$ the following holds:

\beq C\|\omega_S\|_1 < \|\omega_{ S^c}\|_1, \eeq

if for  every  index set $S$ of size $k$,  $\C$ has the
$\text{FCP}(S,C)$, then we say that $\C$ has the fundamental cone
property $\text{FCP}(k,C)$.
\end{definition}

In the next lemma we show how the fundamental cone property can be
used to evaluate the performance of an LP decoder, even when it
fails to recover the true codeword. The key assumption is that the
channel is a bit flipping channel (\textit{e.g.} BSC).

\begin{lem}
\label{lem:LP robustness} Let $\C$ be a code that has the
$\text{FCP}(S,C)$ for some index set $S$ and some $C\geq1$. Suppose
that a codeword $\x^{(c)}$ from $\C$ is transmitted through a bit flipping
channel, and the received codeword is $\x^{(r)}$. If the
pseudocodeword $\x^{(p)}$ is the output of LP decoder for the
received codeword $\x^{(r)}$, then the following holds:

\beq \|\x^{(p)}-\x^{(c)}\|_1 <
2\frac{C+1}{C-1}\|(\x^{(r)}-\x^{(c)})_{S^c}\|_1.
\label{eq:robustness}\eeq
\end{lem}
\begin{proof}
Without loss of generality, we may assume that the all zero codeword
was transmitted, i.e. $\x^{(c)}=0$. We have
\begin{align}
\nonumber\|\x^{(r)}_S\|_1 + \|\x^{(r)}_{S^c}\|_1  &= \|\x^{(r)}\|_1 \\
& \stackrel{(a)}{\geq} \|\x^{(p)}-\x^{(r)}\|_1 \nonumber\\ %\label{eq:a}\\
&= \|(\x^{(p)}-\x^{(r)})_S\|_1+\|(\x^{(p)}-\x^{(r)})_{S^c}\|_1\nonumber\\
& \stackrel{(b)}{\geq} \|\x^{(r)}_S\|_1 -\|\x^{(p)}_S\|_1  + \|\x^{(p)}_{S^c}\|_1 -
\|\x^{(r)}_{S^c}\|_1. \label{eq:b}
\end{align}
\noindent (a) is true because from (\ref{BSC-LP decoder}),
$\|\x^{(p)}-\x^{(r)}\|_1 \leq \|\x^{(c)}-\x^{(r)}\|_1$. Also
(b) holds by the triangular inequality. Note that
$\x^{(p)}\in\mathcal{K}(H)$, so by definition, $C\|\x^{(p)}_S\|_1 <
\|\x^{(p)}_{S^c}\|_1$. This implies that

\beq \|\x^{(p)}_{S^c}\|_1 - \|\x^{(p)}_S\|_1 >
\frac{C-1}{C+1}\|\x^{(p)}\|_1. \eeq

\noindent Applying this to the left hand side of (\ref{eq:b}) we
obtain

\begin{equation}
2\frac{C+1}{C-1}\|\x^{(r)}_{S^c}\|_1 > \|\x^{(p)}\|_1,
\end{equation}
\noindent Which is the desired result.

\end{proof}
\noindent An asymptotic case of Lemma \ref{lem:LP robustness} for
$C\rightarrow 1$ is in fact equivalent to the LP success condition.
Namely, let $S$ be the index set of the flipped bits in the
transmitted codeword, i.e. the set of bits that differ in $\x^{(r)}$
and $\x^{(c)}$. If  $\text{FCP}(S,C)$ holds for some $C>1$, then Lemma \ref{lem:LP robustness} implies that LP
decoding can successfully recover the original codeword. Now let us
say that the set of errors (flipped bits) is slightly larger than
$S$, and does include $S$. Then the vector $(\x^{(r)}-\x^{(c)})_{S^c}$
has a few (but not too many) nonzero entries. Therefore, even if the
LP decoder output $\x^{(p)}$ is not equal to the actual codeword, it
is still possible to obtain an upper bound on its $\ell_1$ distance to
the unknown codeword. We recognize this as the
robustness of LP decoder, and characterize it by
$\text{FCP}(S,C)$, for $C>1$. Furthermore, two notions of robustness can be considered. Strong robustness means that for \emph{every} set $S$ of up to some cardinality $k$, the $\text{FCP}$ condition holds, namely $\text{FCP}(k,S)$. Weak robustness on the other hand deals with almost all sets $S$ of up to
a certain size.  In the next section we present a thorough analysis of LP
robustness for two categories of codes: expander codes and codes
with $\Omega(\log\log n)$ girth. For these two classes of codes, rigorous
analysis has been done on the performance of LP decoders in
\cite{Feldman constant fraction of error,Alex Decoding} and
\cite{ADS}, respectively. We build on the existing arguments to incorporate the robustness condition and analyze the fundamental cone property. Afterwards, we discuss the implications of LP robustness.
%In fact one might even be able to find an upper bound on
%the number of bits that are significantly twisted in $x^{(p)}$
%versus the original codeword $x^{(c)}$, and find out cases for which
%the \emph{rounded pseudocodeword $\langle x^{(p)} \rangle $ has
%significantly less errors than $x^{(r)}$.}

\section{Analysis of LP Robustness} \label{sec:p-qmatchings}
In most cases, if there exists a certificate for the success of LP decoder, it can be
often extended to guarantee that the LP decoder is robust, namely
that the FCP condition is satisfied for some $C > 1$. By
carefully re-examining the analysis of LP decoder, one might be able
to do such a generalization. This is the main focus of this
section. We consider three major methods that exist in the literature for analyzing the performance of LP decoders. The first one is due to Feldman et. al \cite{Feldman
constant fraction of error}, and is based on using a dual witness type of
argument to certify the success of LP decoder for expander
graphs. The second one is that of Daskalakis \textit{et al.} \cite{Alex Decoding},
which again considers linear programming decoding in expander codes.
Specifically, \cite{Alex Decoding} analyzes the dual of LP and finds a simple combinatorial condition for the dual value to be zero (implying that the LP decoder is successful). The condition is
basically the existence of a so-called \emph{hyperflow} from the set of
flipped bits to unflipped bits. The existence of a valid hyperflow can be secured by the
presence of so-called $(p,q)$-matchings.  It then follows from a detailed series of probabilistic calculations     that $(p,q)$-matchings of interest exist for certain expander codes. The main difference between this analysis and
that of Feldman \textit{et al.} is the probabilistic nature of the arguments in \cite{Alex Decoding},  which account for weak recovery thresholds.

A third analysis of the LP decoder was done by Arora \textit{et
al.}, \cite{ADS}, which is based on factor graphs with a doubly logarithmic girth. Unlike previous dual feasibility arguments, the authors in \cite{ADS} introduce a certificate in the primal domain, which is of the following form: If
in the primal LP problem, the value of the objective function for
the original codeword is smaller than its value for all vectors
within a local deviation from the original codeword, then LP decoder
succeeds. Local deviations are defined by weighted minimal local
trees whose induced subgraphs are cycle-free.

\subsection{Strong LP Robustness for Expander Codes}
Strong thresholds of LP decoding for expander codes are derived in
\cite{Feldman constant fraction of error}. To show that the
transmitted codeword is the LP optimal obtained by (\ref{LP decoder}) when a subset of the bits are flipped, a set of
feasible dual variables are found that satisfy the following
conditions. Suppose the factor graph of $\C$ is denoted by
$\mathcal{G} = (X_v,X_c,\mathcal{E})$. We may also assume without
loss of generality that the all zero codeword was transmitted. A
set of feasible dual variables is defined as follows (see
\cite{Feldman constant fraction of error} for more details)

\begin{definition}
\label{def:Dual}
For an error set $S$, a set of feasible dual
variables is a labeling of the edges of the factor graph $\mathcal{G}$,
say $\{\tau_{ij}~|~v_i \in X_v~c_j\in X_c\}$, where the following two
conditions are satisfied:
\begin{itemize}
\item[i)] For every check node $c_j\in X_c$ and every two disjoint
neighbors of $c_j$, say $v_i,v_{i'}\in N(j)$, we have
$\tau_{ij}+\tau_{i'j} \geq 0$.
\item[ii)] For every variable node $v_i \in X_v$, we have $\sum_{c_j\in N(v_i)} \tau_{ij} \leq
\gamma_i$.
\end{itemize}
\end{definition}

We show that a generalized set of dual feasible variables can be
used to derive LP robustness. To this end, we show that the
existence of a set of feasible dual variables implies the FCP
condition.  The following lemma is proved in Appendix \ref{sec:proof of feasibility lem}.
\begin{lem}
\label{lem:feasible-->LP} Suppose that a set of dual variables
satisfy the feasibility conditions (Definition \ref{def:Dual}) for
an arbitrary log-likelihood vector $\gamma$. Then for every vector
$\omega \in \mathcal{K}(\C)$, the following holds \beq \sum_{1\leq i\leq
n}\gamma_i\omega_i > 0.\eeq
\end{lem}

\noindent A special case of Lemma \ref{lem:feasible-->LP} is when
the channel is a BSC, and a set $S$ of the bits have been flipped.
We can also assume without loss of generality that the all zero
codeword was transmitted. Then Lemmas \ref{lem:LP
robustness} and \ref{lem:feasible-->LP} imply that if a dual feasible set exists, then LP
decoder succeeds, which is the conclusion of \cite{Feldman constant
fraction of error}. In this case the log-likelihood vector $\gamma$
takes the value $-1$ over the set $S$ and $1$ over the set
$S^c$. Let us now define a new likelihood vector $\gamma'$ by

\beq \gamma' = \left\{\begin{array}{c} -C ~~ i\in S
\\ 1~~i\in S^c\end{array}\right., \eeq

\noindent for some $C>1$. If a dual feasible set exists that
satisfies the feasibility condition for $\gamma'$, then
it follows that $\text{FCP}(S,C)$ holds. Knowing this and pursuing an argument very similar to \cite{Feldman constant fraction of error} for the construction of dual feasible in expander codes, we are able to prove the following lemma, the proof of which is given in Appendix \ref{sec:proof of strong robust}.   %\cite{Feldman
%constant fraction of error} uses the idea of
%$(\delta,\lambda)$-matchings to show that a dual feasible set exists
%for expander codes. A $(\delta,\lambda)$-matching on a set $S$ is a
%matching from every node in $S$ to $\delta d_v$ distinct check
%nodes, and also from every node is $\dot{S}$ to $\lambda d_v$
%distinct check nodes, where $\dot{S}$ is the set of variable nodes
%that are connected to at least $(1-\lambda)d_v$ check nodes in
%$N(S)$. The edges of this matching are then assigned labels from the
%set $\{x,-x,0\}$ in a simple manner, to assure that the labels of
%the edges adjacent to every node in $S$ sum to less than $-1$ and,
%to every node in $S^c$ sum to less than $1$ (see \cite{Feldman
%constant fraction of error}). This entails that
%$\frac{1}{(2\delta-1)d_v}<x<\frac{1}{(1-\lambda)d_v}$. The key point
%in achieving a robustness condition is the slack in the range of
%possible $x$. By choosing the largest possible $x$, The
%$\text{FCP}(S,C)$ can be maintained for some $C>1$. This is stated
%in the next theorem, the proof of which is omitted for brevity.

\begin{thm}
\label{thm:strong robust} Let $\G$ be the factor graph of a code
$\C$ of length $n$ and rate $R = \frac{m}{n}$, and let $\delta > 2/3
+ 1/d_v$. If $\G$ is a bipartite $(\alpha n, \delta d_v)$ expander graph,
then $\C$ has $\text{FCP}(t,C)$, where $t = \frac{3\delta
-2}{2\delta -1}\alpha$ and $C = \frac{2\delta-1}{2\delta - 1 -1/d_v}
$. This means that for every every set $S$ of size $t$,
$\text{FCP}(t,C)$ holds.
\end{thm}

Basically, \cite{Feldman constant fraction of error} shows that if
the conditions of Theorem \ref{thm:strong robust} are satisfied, then LP
succeeds for every error set of size $t$, namely that $\text{FCP}(t,1)$ holds. However Theorem
\ref{thm:strong robust} asserts that, in addition, a strong
robustness holds, i.e. $\text{FCP}(t,C)$ for some $C>1$.
\subsection{Weak LP Robustness for Expander Codes}

We show that for random expander codes a probabilistic analysis
similar to the dual witness analysis of \cite{Alex Decoding} can be
used to find the extents of the fundamental cone property for expander codes, in a weak
sense. We rely on the matching arguments of
\cite{Alex Decoding},  with appropriate
adjustments. The following definition is given in \cite{Alex
Decoding}.
\begin{definition}
\label{def:p-q matching}
For nonnegative integers $p$ and $q$, and a
set $F$ of variable nodes, a $(p,q)$-matching on $F$ is defined by
the following conditions:
\begin{enumerate}
\item[(a)] each bit $v_i \in F$ must be matched with $p$ distinct check nodes,
and

\item[(b)] each variable node $v_{i'} \in F^c$ must be connected with
\beq X_{i'} := \max\{q-d_v + Z_{i'} , 0\} \eeq  checks nodes from
the set $N(F)$, that are different from the check nodes that the
nodes in $F$ are matched to, where $Z_{i'}$ is defined as $Z_{i'} :=
|N(i') \cap N(F)|$.
\end{enumerate}
\end{definition}

We prove the following lemma that relates the existence of a
$(p,q)$-matching to the fundamental cone property of a code $\C$.
This lemma is proved in Appendix \ref{sec:proofs of lemma matching}.
\begin{lem}
\label{lem:matching} Let $\C$ be a code of rate $R$ with a bipartite
factor graph $\G$, where every variable node has degree $d_v$. Let
$S$ be a subset of the variable nodes of $\G$. If a $(p,q)$-matching
on $S$ exists, then $\C$ has the $\text{FCP}(S,
\frac{2p-d_v}{d_v-q}). $
\end{lem}

\cite{Alex Decoding} provides a probabilistic tool for the
existence of $(p,q)$-matchings in regular bipartite expander
graphs, which helps answer the question of how large an error set LP
decoding can fix. For example, for a random LDPC(8,16)  code, the
probabilistic analysis implies that with high probability, a fraction $0.002$ of errors
is recoverable using LP decoder. However, taking the specifications
of the matching that leads to this conclusion and applying Lemma
\ref{lem:matching}, it turns out that for an error set of size
$0.002n$, the robustness factor is at least $C=1.3$, i.e the code
has $\text{FCP}(0.002n,1.3)$.

\subsection{Weak LP Robustness for Codes with $\Omega(\log\log(n))$ Girth}
Recall that $\G = (X_v,X_c,\mathcal{E})$ is used to denote the
factor graph of the parity check matrix $H$ (or of code $\C$), where
$X_v$ and $X_c$ are the sets of variable and check nodes
respectively and $\mathcal{E}$ is the set of edges. Also recall that
the girth of $\G$ is defined as the size of the shortest cycle in
$\G$. Without loss of generality, we assume that $X_v =
\{v_1,v_2,\cdots,v_n\}$, where $v_i$ is the variable node
corresponding to the $i^\text{th}$  bit of the codeword.  Let $T\leq
\frac{1}{4} \text{girth}(\mathcal{G})$ be fixed. The following notions are
defined in \cite{ADS}.
\begin{definition}
A tree $\mathcal{T}$  of height $2T$ is called a skinny subtree of
$\G$, if it is rooted at some variable node $v_{i_0}$,  for every
variable node $v$ in $\mathcal{T}$ all the neighboring check nodes
of $v$ in $\G$ are also present in $\mathcal{T}$, and for every
check node $c$ in $\mathcal{T}$ exactly two neighboring variable
nodes of $c$ in $\G$ are present in $\mathcal{T}$.
\end{definition}

\begin{definition}
Let $\w\in[0,1]^T$ be a fixed vector. A vector $\beta^{(\w)}$ is
called a minimal $T$-local deviation, if there is a skinny subtree
of $\G$ of height $2T$, say $\mathcal{T}$, so that for every
variable node $v_i~1\leq i \leq n$, \beq \nonumber
\beta_i^{(\w)}=\left\{\begin{array}{l}\w_{h(i)}~~~~\text{if
$v_i\in \mathcal{T}\setminus\{v_{i_0}\}$} \\
0~~~~\text{otherwise}\end{array}\right.,\eeq

\noindent where $h_i = \frac{1}{2}d(v_{i_0},v_i)$.
\end{definition}

The key to the derivations of \cite{ADS} is the following  lemma:
\begin{lem}[Lemma 1 of \cite{ADS}]
\label{lem:Exp} For any vector $\z\in\mathcal{P}$, and any positive
vector $\w\in[0,1]^T$, there exists a distribution on the minimal
$T$-local deviations $\beta^{(\w)}$, such that \beq \nonumber
\stexp\beta^{(\w)} = \alpha \z, \eeq \noindent where $0<\alpha\leq 1$.
\end{lem}

%\begin{figure}[t]
%\centering
%  \includegraphics[width= 0.5\textwidth]{images/C_vs_p_lemma9.eps}
%  \caption{ \scriptsize{Robustness factor $C$ as a function of error probability $p$ for $d_c = 6$ and $d_v = 3$, based on Theorem \ref{thm:ADS robust 1}.}}
%  \label{fig:C_vs_P1}
%\end{figure}
\noindent Lemma \ref{lem:Exp} has the following interpretation. If a
linear property holds for all minimal $T$-local deviations (\textit{e.g.}
$f(\beta^{(\w)}) \geq 0$, where $f(.)$ is a linear operator), then it also holds
for all pseudo-codewords (i.e. $f(\z) \geq 0~\forall \z\in\mathcal{P}$).
Interestingly enough, the robustness of LP decoding for a given set of
bit flips $S$ has a linear certificate, namely
$\text{FCP}(S,C)$\footnote{Note that this is only true for bit fliping
channels, where the output alphabet in the binary field.}. In other words, if we define:

\beq \nonumber f_C^{(S)}(\x) = \sum_{i\in S^c} x_i -C\sum_{i\in S}
x_i, \eeq

%
%\begin{figure}[t]
%\centering
%  \includegraphics[width= 0.5\textwidth]{images/C_vs_p_lemma10.eps}
%  \caption{ \scriptsize{Robustness factor $C$ as a function of error probability $p$ for $d_c = 6$ and $d_v = 3$, based on Theorem \ref{thm:ADS robust 2}.}}
%  \label{fig:C_vs_P2}
%\end{figure}
\noindent then $\text{FCP}(S,C)$ holds, if and only if
$f_1^{(S)}(\z)\geq 0$ for every pseudocodeword  $\z\in \mathcal{P}$.
Therefore, according to Lemma \ref{lem:Exp}, it suffices that the
condition be true for all $T$-local deviations. Furthermore, for
arbitrary $C>1$, if $f_C^{(S)}(\beta^{(\w)})\geq 0$ for all minimal
$T$-local deviations $\beta^{(\w)}$, then it follows that the code
has the $\text{FCP}(S,C)$ property. This simple observation helps us
extend the probabilistic analysis of \cite{ADS} to  robustness
results for LP decoding. The resulting key theorem is mentioned
below, the proof of which can be found in Appendix \ref{sec:proofs of thms}. In order to state the theorem, first we define  $\eta$ to be a
random variable that takes the value $-C$ with probability $p$ and
value $1$ with probability $1-p$. Also, define the sequences of random
variables $X_i,Y_i,~ i\geq 0$, in the following way:

\bea \nonumber Y_0 &=& \eta, \\
     \nonumber X_i &=& \min\{Y_i^{(1)},\dots,Y_i^{(d_c-1)}\}~~\forall i>0, \\
     \nonumber Y_i &=&
     2^i\eta+X_{i-1}^{(1)}+\cdots+X_{i-1}^{(d_v-1)}~~\forall i>0, \\
\eea \noindent Where $X^{(j)}$s are independent copies of a random
variable $X$.

\begin{thm}
\label{thm:ADS robust 3}
 Let $0\leq p \leq 1/2$ be the probability
of bit flip, and $S$ be the random set of flipped bits.  If for some
$j\in \mathbb{N}$,  \begin{eqnarray*}
c = \gamma^{1/(d_v-2)}\min_{t\geq 0}\stexp{e^{-tX_j}}<1,
\end{eqnarray*}
\noindent where $\gamma =  (d_c-1)\frac{C_{R}+1}{C_{R}}(\frac{C_{R} \cdot p}{1-p})^{1/(C_{R}+1)}(1-p)<1$, Then with probability
at least $1-O(n)c^{d_v(d_v-1)^{T-1}}$  the code $\mathcal{C}$ has the
$\text{FCP}(S,C)$, where   $T$ is any integer with $j \leq T<
1/4\text{girth}(\mathcal{G})$.
\end{thm}

%\begin{figure}[t]
%\centering
%  \includegraphics[width= 0.5\textwidth]{images/C_vs_p_lemma10.eps}
%  \caption{ \scriptsize{Approximate upper bound for the robustness factor $C$ as a function of error probability $p$ for $d_c = 6$ and $d_v = 3$, based on Theorem \ref{thm:ADS robust 3}.}}
%  \label{fig:C_vs_P2}
%\end{figure}
For $d_c = 6$ and $d_v = 3$, a lower bound on the robustness parameter $C$ that results from Theorem \ref{thm:ADS robust 3} is plotted against the probability of bit flip $p$, in Figure \ref{fig:C_vs_P3}.

\section{Implications of LP robustness}

\subsection{Mismatch Tolerance}

One of the direct consequences of the robustness of LP decoding is
that if there is a slight mismatch  in the implementation of the LP
decoder, its performance does not degrade significantly.
More formally, suppose that due to noise, quantization or some other factor, a
mismatched log-likelihood vector $\gamma' = \gamma + \Delta\gamma$
is used in the LP implementation. We refer to such a decoder as a
\emph{mismatched LP decoder}. Since the channel is BSC, the entries
of $\gamma$ all have the same amplitude $g$. We also define $\delta
= \max_{i}|\Delta\gamma_i|$, and assume that $\delta < g$. We can
prove the following theorem.

\begin{thm}
\label{thm:mismatch} Suppose that $S$ is the set of bit errors. Let
$C = \frac{g+\delta}{g-\delta}$. If $\C$ has $\text{FCP}(S,C)$, then
the mismatched LP decoder corrects all errors and recovers the original codeword.
\end{thm}

\begin{proof}
We assume without loss of generality that the all zero codeword is
transmitted. We show that if $\text{FCP}(S,C)$ holds, then the all
zero codeword is the minimum cost vector in the polytope
$\mathcal{P}$. Suppose $\omega$ is a nonzero vector in
the fundamental code $\mathcal{K}$. We begin with the definition of
$\text{FCP}(S,C)$ and write
 \bea -C\sum_{i\in S}\omega_i + \sum_{i\in S^c}\omega_i >0.\eea
\noindent Multiply both sides by $(g-\delta)$:

 \bea -\sum_{i\in S}(g+\delta)\omega_i + \sum_{i\in S^c}(g-\delta)\omega_i
> 0.\eea
\noindent We also know from the definition of $\delta$ that
$\gamma'_i
>(g-\delta) $ for $i\in S^c$, and $\gamma'_i > -g - \delta$ for $i\in
S^c$, and that $\omega\geq 0$. Therefore
 \bea -\sum_{i\in S\cup S^c}\gamma'_i\omega_i > 0,\eea
\noindent which proves that the all zero codeword is the unique
minimum cost solution of the mismatched LP.

\end{proof}

\subsection{Pseudocodewords and High Error Rate Subsets}
\label{sec:HER subset} We showed in Section \ref{sec:robutsness}
that for an appropriate code $\C$, even when LP decoder fails to
recover an actual codeword from the output of a BSC, the $\ell_1$
distance between the obtained pseudocodeword and the actual codeword
can be bounded by a finite factor of excess errors (see equation \ref{eq:robustness}). We now show that this property allows us to use the output of LP decoder to find a \emph{high error} rate subset of the bits of linear size, namely a subset of
bits over which the fraction of errors is significantly larger than the fraction of errors in the
entire received codeword. Obtaining  such  \emph{importance} subset is very crucial, since it provides additional information about a significant proportion of the bits which can be used to improve the decoder's performance. For instance, one can impose additional soft or hard
constraints on the importance subset, and solve a constrained
linear program or other post processing algorithms following the  initial linear program. This forms the
idea for the proposed iterative LP decoding algorithm which will be outlined in Section \ref{sec:Algorithm}.

Consider a code $\C$ of length $n$ and rate $R$, and a codeword
$\x^{(c)}$ from $\C$ transmitted through a bit flipping channel. Suppose that a
set $K$ of the bits get flipped, where the cardinality of $K$ is
$(1+p^*)\epsilon n$ for some $0<p^*<1$ and $\epsilon>0$. Denote the
received vector by $\x^{(r)}$. We are interested in the case where LP
fails, so the LP minimal $\x^{(p)}$ is a fractional pseudocodeword.
However, the size of the error set is only slightly larger than the
correctable size $p^*n$. In other words, we assume that for some subset
$K_1\subset K$ of size $p^*n$, the code has $\text{FCP}(K_1,C)$,
for some $C>1$. We show in the next lemma that the
index set of the largest $k$ entries of the vector $\x^{(r)}-\x^{(p)}$ has a significant overlap with $K$ with high probability, and is thus a high error rate subset of entries. The following theorem formalized this claim.

% When $n$ is large,
%the error recovery of LP decoder exhibits a phase transition
%behavior. In other words, we assume that there is a critical value
%$0<p^*<1$, so that if $p<p^*$, then with high probability LP decoder
%succeeds, otherwise ($p>p^*$) fails.
\begin{figure}
\centering
  \includegraphics[width= 0.4\textwidth]{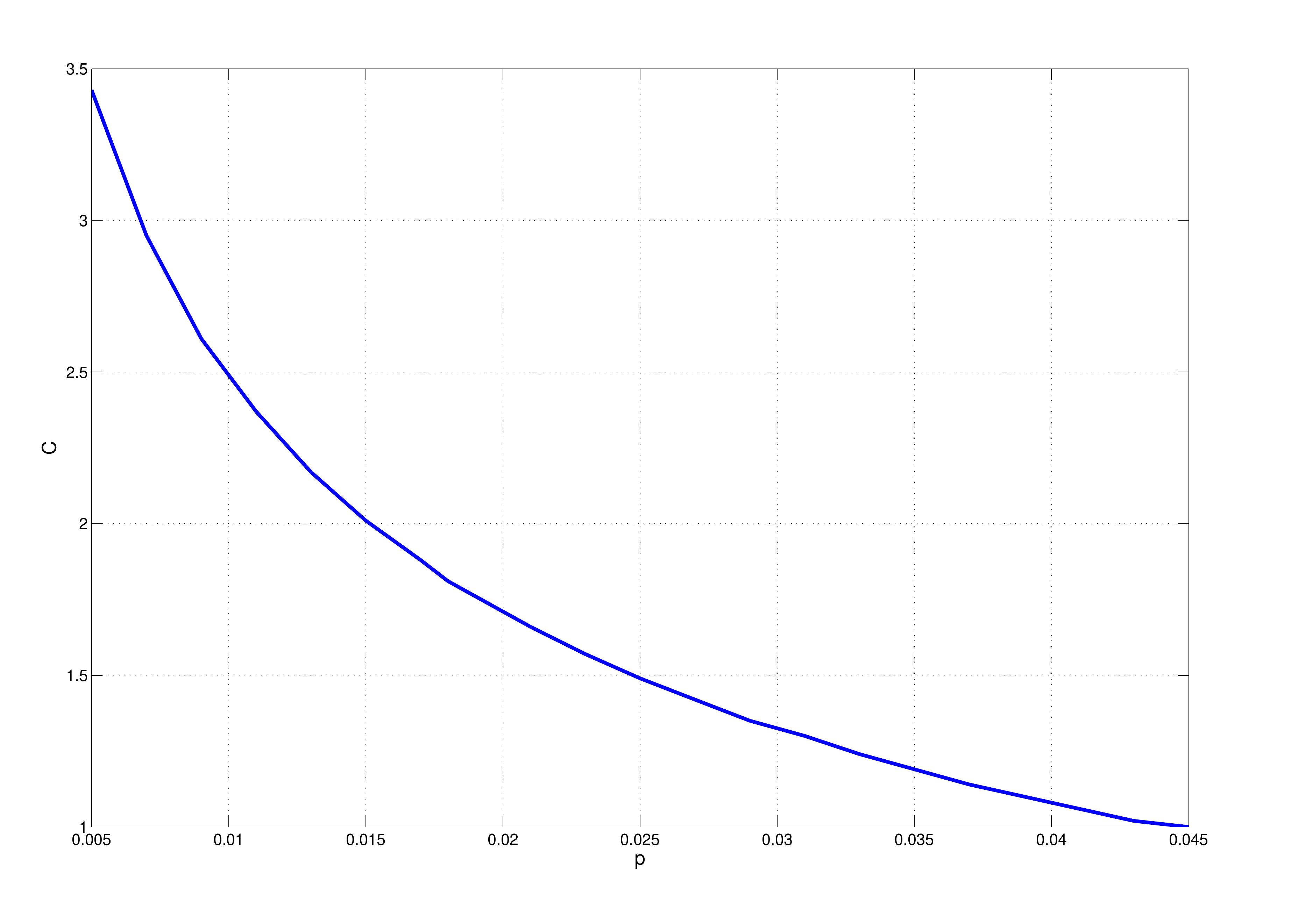}
  \caption{ \scriptsize{Approximate upper bound for the robustness factor $C$ as a function of error probability $p$ for $d_c = 6$ and $d_v = 3$, based on Theorem \ref{thm:ADS robust 3}.}}
  \label{fig:C_vs_P3}
\end{figure}

\begin{thm}
\label{thm:importance set} Suppose that a codeword $\x^{(c)}$ is
transmitted through a bit flipping channel, and the output $\x^{(r)}$ differs from the
input in a set $K$ of the bits with $|K| = p^*(1+\epsilon)n$, for
some $0<p^*<1$ and $\epsilon>0$. Also, suppose that for a subset
$K_1\subset K$ of size $p^*n$, $\text{FCP}(K_1,C)$ holds, for
some $C>1$, and that the LP minimal is the pseudocodeword $\x^{(p)}$. If $L$ is the set of the $p^*(1+\epsilon)n$ largest entries of the
vector $\x^{(r)}-\x^{(p)}$ in magnitude, then the fraction of errors
in $\x^{(r)}$ over the set $L$ is at least
$1-2\frac{C+1}{C-1}\epsilon $.
\end{thm}

Before proving this theorem, we state the following definition and
lemma.
\begin{definition}
Let $\x\in\mathbb{R}^n$ be a $k$-sparse vector. For $\lambda > 0$, We
define $W(x,\lambda)$ to be the size of the largest subset of
nonzero entries of $\x$ that has a $\ell_1$ norm less than or equal
to $\lambda$, i.e., \beq W(\x,\lambda) := \max\{|S|~|~S\subseteq
supp(\x),~\|\x_S\|_1\leq \lambda\}.\eeq
\end{definition}

\noindent The following Lemma is proven in \cite{Iterative l1 ISIT}.
\begin{lem}[Lemma 1 of \cite{Iterative l1 ISIT}]
Let $\x$ be a $k$-sparse vector and $\hat{\x}$ be another vector.
Also, let $K$ be the support set of $\x$ and $L$ be the $k$-support
set of $\hat{\x}$, namely the set of $k$ largest entries of
$\hat{\x}$. If $d = \|\x-\hat{\x}\|_1$, then \beq |K\cap L|\geq
k-W(\x,d).  \eeq \label{lem:deviation thm}
\end{lem}
%
%\begin{figure*}[t]
%\centering
%  \includegraphics[width= 0.6\textwidth]{images/HERLERoverlap.eps}
%  \caption{ \scriptsize{The average fraction or (empirical) probability of bit flip when LP fails, in a fraction $\alpha$ of the  bits based on the following
%  criteria; The most and the least deviated bits in the LP minimal  from the received binary vector. The most and least
%  certain bits  (closest to and farthest from $0.5$) in the LP minimal solution. The code is a random LDPC(3,4) with $n=200$
%  and the original channel error probability is $p=0.15$.}}
%  \label{fig:HERLEROverlap}
%\end{figure*}

\noindent\begin{proof}[Proof of Theorem \ref{thm:importance set}]

Define $k= p^*(1+\epsilon)n$, and apply Lemma \ref{lem:deviation
thm} to the $k$-sparse vector $\x^{(r)}-\x^{(c)}$, and the vector
$\x^{(p)}-\x^{(r)}$. If $L$ is the index set of the largest $k$
entries of $\x^{(p)}-\x^{(r)}$ in magnitude, then from Lemma
\ref{lem:deviation thm} we have

\beq \label{eq:KcapL}|K\cap L| \geq k - W(\x^{(r)}-\x^{(c)},\Delta),
\eeq

\noindent where $\Delta = \|\x^{(c)}-\x^{(p)}\|_1$. Since
$\|\x^{(r)}-\x^{(c)}\|$ has only $\pm1$ nonzero entries,
(\ref{eq:KcapL}) can be written as

\beq \label{eq:KcapL2}|K\cap L| \geq k -  \|\x^{(c)}-\x^{(p)}\|_1. \eeq

\noindent We use the inequality in (\ref{eq:robustness}) to further
lower bound the right hand side of (\ref{eq:KcapL2}). Recall that
$K_1\subset K$ is such that $\C$ has $\text{FCP}(K_1,C)$. Therefore,
we can write:

\bea \label{eq:KcapL3}|K\cap L| &\geq& k -2\frac{C+1}{C-1}
\|(\x^{(r)}-\x^{(c)})_{ K_1^c}\|_1    \\ &=& k - 2\frac{C+1}{C-1}
(k-p^*n). \eea

Dividing both sides by $|K| = k$, we conclude that at least a fraction
$1-2\frac{C+1}{C-1}\epsilon$ of the set $L$ are flipped bits.
\end{proof}

\section{Iterative Reweighted LP Algorithm and Improved Strong Threshold}
\label{sec:Algorithm}

First, we briefly define different recovery thresholds for LP decoding for more clarity of the statements that will follow. In general,
the actual weak and strong thresholds for a given classes of linear codes  might be unknown, and the existing threshold only provide lower bounds on these quantities. For expander codes
for instance, the size of the error set that can be recovered via LP
 can be lower bounded by the size of the set for which a dual witness
 exists~\cite{Feldman constant fraction of error,Alex Decoding}. Since a dual witness is only a sufficient condition for the
 success of LP decoding, the actual thresholds are generally expected to be higher. However, to date, the best achievable thresholds for LP
 decoding for expander codes are those given by the dual feasibility arguments. Therefore, we also consider thresholds associated with
 those limits, namely the ``provable'' thresholds. Specifically, we define the following four thresholds
 for LP decoding on a given code $\C$ that has regular variable and check degrees $d_v$ and $d_c$.

\begin{definition}[Recovery thresholds]
\label{def:thresholds}
Strong recovery threshold is denoted by
$p^*_s$, and is defined as the largest fraction such that
\emph{every} set of size $p^*_sn$ is recoverable via LP decoding.
Weak recovery thresholds is denoted by $p^*_w$, and it means that
\emph{almost all} sets of size $p^*_wn$ is recoverable via LP. We
define $p^*_{sd}$ to be the maximum provable strong threshold achieved by a dual feasible,
\cite{Feldman constant fraction of error}. Similarly, $p^*_{wd}$ is the provable weak
threshold, i.e. for almost all sets of size $p^*_{wd} n$, a dual feasible (\cite{Alex Decoding}) exist.
\end{definition}

As sketched in Theorem \ref{thm:importance set}, by examining the deviation of the LP optimal (pseudo-codeword) and
the received vector, it is possible to identify a high error rate (HER) subset of bits in which the fraction of bit flips is higher than the overall probability of error, or the fraction of errors in the complement of the HER set. One way this imbalancedness can be exploited is by using a weighted LP scheme. This is outlined in the following iterative algorithm.
\begin{alg}
\label{alg:3} \text{}
\begin{enumerate}
\item Run LP decoding. If the output is integral terminate,
otherwise proceed.
\item  Take the fractional pseudocodeword $\x^{(p)}$ from the LP
decoder, and construct the deviation vector $\x^{(d)} =
\x^{(r)}-\x^{(p)}$.
\item Sort the entries of
$\x^{(d)}$ in terms of absolute value, and denote by $L$  the index
set of its \textbf{smallest} $pn$ entries.
\item  solve the following weighted LP:
\beq   \min_{\x\in \mathcal{P}} \lambda_1\|(\x-\y)_L\|_1 +
\lambda_2\|(\x-\y)_{L^c}\|_1,\eeq 
\noindent where $\lambda_1$ and
$\lambda_2$, where $\lambda_1 < 0$ ad $\lambda_2 > 0$ are fixed
parameters.
\end{enumerate}
\end{alg}

Algorithm \ref{alg:3} is only twice as complex as LP decoding. We  prove in the following that
algorithm \ref{alg:3} has a strictly improved provable strong and weak
recovery thresholds than the dual feasibility thresholds
$p^*_{sd}$ and $p^*_{wd}$ (Recall the definitions of $p^*_{sd}$ and $p^*_{sd}$ from Definition
\ref{def:thresholds}). %$p^*_{sd}$ is basically the provable strong
%threshold of LP decoding based on the analysis of \cite{Feldman
%constant fraction of error}.
%Namely, for every error set $S$ of size
%less than or equal to $p^*_sn$, a feasible dual of Definition
%\ref{def:Dual} exists for the corresponding log-likelihood vector
%$\gamma$.  Then we can prove the following theorem.
\begin{thm}
\label{thm:reweighted improves}
 For any code $\mathcal{C}$, there exist  $\epsilon_1 >0$,$\epsilon_2> 0$, $\lambda_1 < 0 $ and $\lambda_2 >0$ so that every error set of size
 $(1+\epsilon_1)p^*_{sd}$, and almost all error sets of size $(1+\epsilon_2)p^*_{wd}$ can be corrected by Algorithm \ref{alg:3}. %Furthermore, if $(C-1)/\epsilon$ is large enough, where $C$ is the robustness factor for sets of size $(1-\epsilon)p^*_sn$ (or  $(1-\epsilon)p^*_wn$) , then the strong (weak) threshold on Algorithm \ref{alg:3} is at least $(1+\epsilon)p^*_s$ (or $(1+\epsilon)p^*_w$).
\end{thm}

we start with the following lemma
\begin{lem}
\label{thm:weighted LP }
Suppose a codewords $\x$ transmitted is  through a binary channel. Also suppose
that the bits of $\x$ can be divided into two sets $L$ and $L^c$, so
that at least a fraction $p_1$ of the bits in $L$ are flipped, and
at most a fraction $p_2$ of the bits in $L^c$ are flipped. Then the following
weighted LP decoding
\beq \label{eq:weighted LP} \min_{\x\in \mathcal{P}}
-\|(\x-\y)_L\|_1 + \|(\x-\y)_{L^c}\|_1,\eeq
\noindent can recover $\x$,
provided that \beq \label{eq:p1 p2 relation} (1-p_1)|L| + p_2|L^c|
\leq p^*_{sd}.\eeq
\end{lem}
\begin{proof}
We assume without loss of generality that the all zero codeword has
been transmitted and prove that there exists a feasible dual
(Definition \ref{def:Dual} ) for the LP decoder~\ref{eq:weighted LP}. The
feasible dual must satisfy condition (i) of Definition \ref{def:Dual} for all check nodes, and in
addition: \beq\sum_{j\in N(i)}\tau_{ij} \leq \left\{\begin{array}{l}
~~1~~i\in L\cap S \\ -1~~i\in L\cap S^c \\ -1~~i\in L^c\cap S \\
~~1~~i\in L^c\cap S^c
\end{array}\right..
\label{eq:sum_tau}
\eeq
One can note that the conditions of (\ref{eq:sum_tau}) are equivalent to  $\tau_{ij}$'s being a
feasible dual set for ordinary LP decoder when the error set is $S_1
= (L\cap S^c )\cup(L^c\cap S)$. Therefore if the size of $S_1$ is
smaller than $p^*_{sd} n$, from the definition of $p^*_{sd}$, such a
feasible dual set exists. This completes the proof the theorem.
\end{proof}

\begin{proof}[proof of Theorem \ref{thm:reweighted improves}]
We set $\lambda_1 = -1$ and $\lambda_2 = 1$. Suppose the all zero
codeword have been transmitted without loss of generality, and the
received binary vector $x^{(r)}$ has $pn$ errors, where $p =
(1+\epsilon_0)p^*_{sd}$. From Theorem \ref{thm:strong robust}, $\C$
has $\text{FCP}(p^*_sn,C)$ for some $C>1$. Therefore, if we apply
Theorem \ref{thm:importance set} to the output of LP, namely
$\x^{(p)}$, we conclude that the set $L$ of most $pn$ deviated bits
in $\x^{(p)}$ with respect to $\x^{(r)}$, and the set $S$ of the
errors in $\x^{(r)}$, have at least a fraction
$1-2\frac{C+1}{C-1}\epsilon_1$ overlap. Define $p_1 =
\frac{|L\cap S|}{|L|}$ and $p_2 = \frac{|L^c\cap S|}{|L^c|}$. We
 must have 
\bea p_1 &\geq& 1-2\frac{C+1}{C-1}\epsilon_0, \\ p_1|L|
+ p_2|L^c| &=& p. \eea 

\noindent Therefore, as $\epsilon_0\rightarrow
0$, $p_1\rightarrow1$ and $p_2\rightarrow0$. So, for some small
enough $\epsilon_0$, the following will eventually hold

\beq (1-p_1)|L| + p_2|L^c| \leq p^*_{sd}.\eeq

\noindent Thus, according to Lemma \ref{thm:weighted LP }, the
weighted LP step of Algorithm \ref{alg:3} corrects all errors.
similarly, if  a random set of  $pn$ bits are flipped, when $p =
(1+\epsilon_2)p^*_{wd}$, from Lemma \ref{lem:matching} we conclude
that with high probability there exists a $C>1$ so that
$\text{FCP}(S_1,C)$ holds for a random subset $S_1$ of the bit errors of
size $p^*_{wd}n$. Therefore, using Theorem \ref{thm:importance set},
it follows that the set $L$ of most $pn$ deviated bits in $\x^{(p)}$
with respect to $\x^{(r)}$, and the set of errors in $\x^{(r)}$
have at least an overlap fraction of $1-2\frac{C+1}{C-1}\epsilon_2$. The
remainder of the proof is the same as the previous case, i.e. by applying
Lemma \ref{thm:weighted LP }.
\end{proof}

\section{Simulations}
\label{sec:simulations} We have implemented Algorithm \ref{alg:3}
on a random LPDC code of size $n=1000$ and rate $R=3/4$ and have
compared the results with other existing methods. The variable node
degree is $d_v = 3$, and thus, $d_c=4$.  The algorithm is compared
with the mixed integer method of Draper and Yedidia \cite{Mixed
Integer Yedidia}, and the random facet guessing algorithm of
\cite{facet guessing}. The mixed integer algorithm  re-runs
the LP decoding by setting integer constraints on a small subset of ``least
certain'' bits, namely the positions where the LP minimal
pseudocodeword entries are closest to $0.5$. We have  taken the size
of the constrained subset to be $M=5$, which means the number of
extra iterations is $32$ for the mixed integer method.  We also choose to run $20$ more extra
random iterations for facet guessing. In random facet guessing, a
face (facet) of the polytope $\mathcal{P}$ is selected at random, among all
the faces on which the LP minimal pseudocodeword does not reside. Then, LP
decoder is re-run with the additional constraint that the solution
is on the selected face. In contrast, Algorithm \ref{alg:3} has only
one extra iteration.  All methods are simulated in MATLAB where LP
decoder is implemented via the cvx toolbox \cite{cvx}. We have
plotted the BER curves versus the probability of error $p$ in Figure
\ref{fig:BER_LP_vs_WLP}. For Algorithm \ref{alg:3}, for each $p$, we
have experimentally found the optimal $\lambda_1$ and $\lambda_2$ by
choosing the values that on average result in the best performance.
For most of the cases the chosen values where in the ranges $-3\leq
\lambda_1 \leq -0.5$ and $1\leq \lambda_2\leq 3$. Observe the superior BER performance of Algorithm \ref{alg:3} which becomes more significant for smaller values of $p$. For $p=0.11$, the BER improvement in the reweighted LP method is at least one order of magnitude.  In our preliminary experimental evaluation we observe that the BER
curves eventually collapse into the same curve as the LP curve, except for the reweighted LP algorithm, which is an indication of the fact that the empirical thresholds of Algorithm \ref{alg:3} are better than those of LP decoder and existing polynomial time post processing methods.
\begin{figure*}[t]
\centering
  \includegraphics[width= 0.7\textwidth]{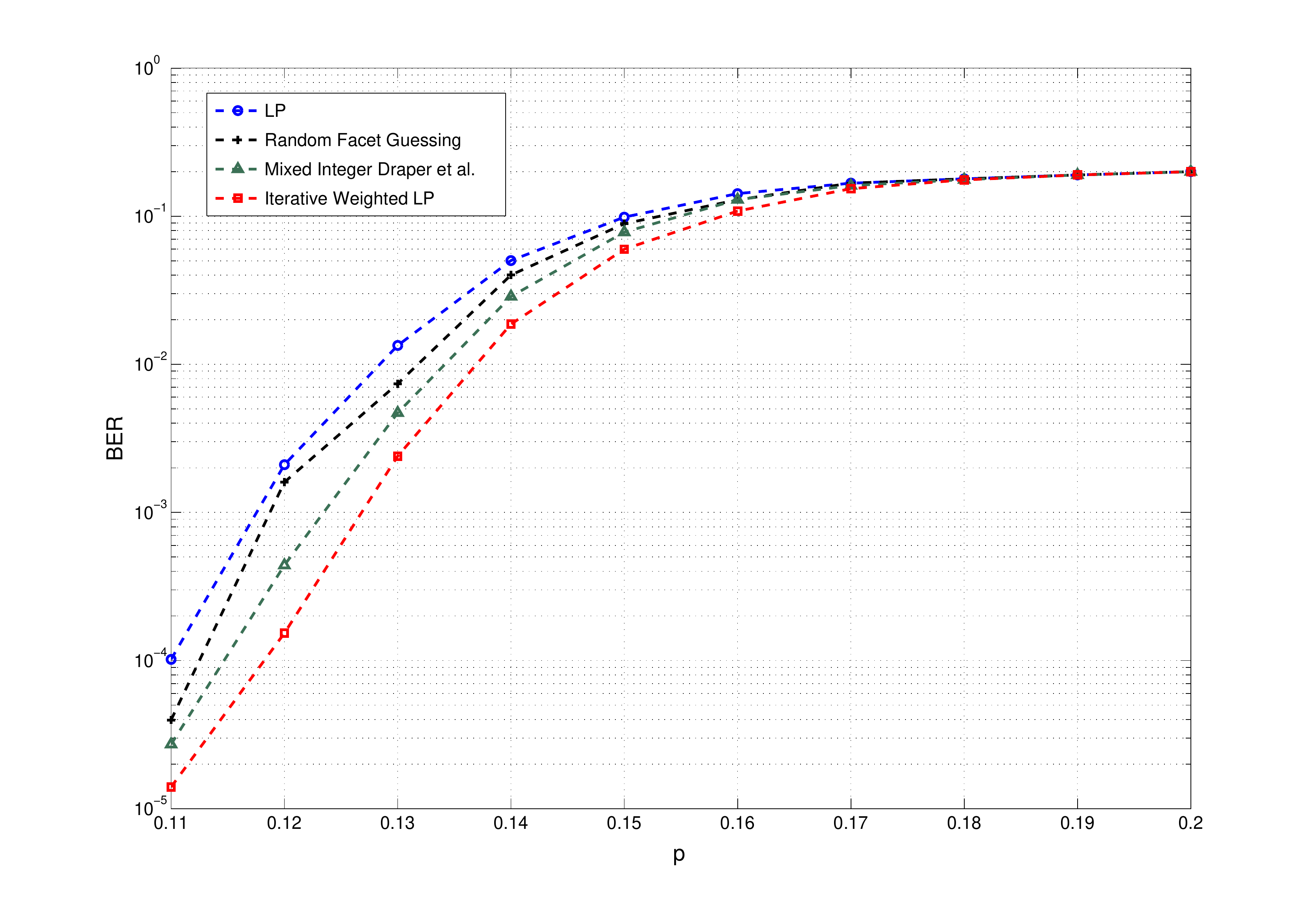}
  \caption{ \scriptsize{BER curves as a function of channel flip probability $p$, for LP decoding and different iterative
  schemes; random facet guessing of \cite{facet guessing}, mixed integer
  method of \cite{Mixed Integer Yedidia}, and  the suggested iterative reweighted LP of Algorithm \ref{alg:3}.
  The code is a random LDPC(3,4) of length $n=1000$.}}
  \label{fig:BER_LP_vs_WLP}
\end{figure*}

\appendix%
\section{  Proof of Lemma \ref{lem:feasible-->LP}  }\label{sec:proof of feasibility lem}
We first prove the following lemma.
\begin{lem}
\label{lem:wtau>0} Suppose $\{\tau_{ij}~|~1\leq i\leq n, 1\leq j\leq
m\}$ is a set of feasible dual variables on the edges of the factor
graph $\mathcal{G}$ of the code $\C$, for some arbitrary
log-likelihood vector $\gamma$. Then for every vector $\w \in
\mathcal{K}(\C)$ and every check node $c_j$, the following holds
\beq \sum_{v_i\in N(c_j)} w_i\tau_{ij} \geq 0. \eeq
\end{lem}
\begin{proof}
We only use condition (i) of a feasible set of dual variables. Note
that among the variable nodes in $N(c_j)$, there can be at most one
node $v_i$ with $\tau_{i,j} <0$. Let $v_i$ be such a variable node.
From the definition of $\mathcal{K}$ we can write

\beq \nonumber w_i \leq \sum_{i'\in N(j)\setminus i} w_{i'}, \eeq

\noindent or equivalently:

\beq\label{eq:sum wtau >0} \tau_{ij}w_i + \sum_{v_{i'}\in
N(v_j)\setminus i}|\tau_{ij}| w_{i'} \geq 0. \eeq

\noindent Moreover, we know that $\tau_{ij}+\tau_{i'j} \geq 0$ for
$i'\neq i$, from the condition (i) of the dual feasibility.
Therefore, replacing $\tau_{i'j}$ with $|\tau_{ij}|$ for each
$i'\neq i$ does not decrease the left hand side of (\ref{eq:sum wtau
>0}), and thus \beq \nonumber \sum_{v_i\in N(c_j)} w_i\tau_{ij}
\geq 0. \eeq

\end{proof}

We now invoke Lemma \ref{lem:wtau>0} that for every check node
$c_j$, $\sum_{v_i\in N(c_j)} w_i\tau_{ij} \geq 0$. If we sum these
inequalities for all check nodes $c_j$ we obtain:

\beq \nonumber \sum_{c_j\in X_c}{\sum_{v_i\in N(c_j)} w_i\tau_{ij}}
= \sum_{v_i\in X_v}{w_i\sum_{c_j\in N(v_i)}\tau_{ij} } \geq 0, \eeq

\noindent When $X_v$ and $X_c$ are the sets of variable and check
nodes respectively. Since $\tau_{ij}$s are feasible dual variables,
from condition (ii) of feasibility (Definition \ref{def:Dual}), we
must have $\sum_{c_j\in N(v_i)}\tau_{ij}  < \gamma_i$. It then
follows that

\beq \nonumber \sum_{v_i\in X_v}{\gamma_i w_i } > 0.\eeq

\section{Proof of Theorem \ref{thm:strong robust} }\label{sec:proof of strong robust}
 We basically repeat the argument of \cite{Feldman constant fraction of
 error} with some slight adjustments. Let $S$ be the set of flipped
 bits, or interchangeably the set of corresponding variable nodes in the factor graph
 $\mathcal{G}$ (we use $v_i$ to refer to the variable node corresponding to the $i^\text{th}$
 bit).
\begin{definition}[$(\delta,\lambda)$ matching from \cite{Feldman constant fraction of error}]
\label{def:delta lambda}
 A $(\delta,\lambda)$ matching of the set $S$ is a set $M$ of edges of the
 factor graph $\mathcal{G}$, so that no two edges are connected to the same check
 node, every node in $S$ is connected to at least $\delta d_v$ edges
 of $M$, and every node in $S'$ is connected to at least $\lambda
 d_v$ edges of $M$. Here $S'$ is the set of variable nodes that are
 connected to at least $(1-\lambda)$ check nodes in $N(S)$.
\end{definition}

\noindent If there is a $(\delta,\lambda)$ matching on the set $S$,
then we consider the following labeling of the edges of
$\mathcal{G}$. For a check node $v_j$, if it is adjacent to an edge
$\tau_{ij}$ is $M$ then set $\tau_{ij}=-x$ and $\tau_{i'j}=x$ for
every other variable node $v_i'\in N(v_j)~i'\neq i$. Otherwise,
label all of the edges of the edges adjacent to $j$ by 0. It can be
seen that this for this labeling  $\{\tau_{ij}\}$ satisfies
condition (i) of dual feasibility (Definition \ref{def:Dual}), and
furthermore:  \beq\sum_{j\in N(i)}\tau_{ij} \leq
\left\{\begin{array}{l} (1-2\delta)d_v x~~i\in S \\ (1-\lambda)d_v
x~i\in S^c
\end{array}\right..\eeq

We know take $\lambda = 2-2\delta+1/d_v$. Let us  define a new
likelihood vector $\gamma'$ by

\beq \gamma' = \left\{\begin{array}{c} -C ~~ i\in S
\\ 1~~i\in S^c\end{array}\right.. \eeq

\noindent If a dual feasible set exists that satisfies the
feasibility condition for the vector $\gamma'$, then this implies
that the $\text{FCP}(S,C)$ holds. Now, since $C <
\frac{2\delta-1}{1-\lambda}$, if we choose $x$ to be 
\beq x =\frac{1}{(1-\lambda)d_v},\eeq 
\noindent then, it is clear that
$(1-2\delta)d_vx <-C$. So the dual feasibility condition is
satisfied, if we can construct the required $(\delta,\lambda)$
matching for $S$. From \cite{Feldman constant fraction of error}, if
$|S|\leq \frac{3\delta-2}{2\delta-1}\alpha$, and $\G$ is a bipartite
$(\alpha n,\delta d_v)$ expander, the desired matching exists. This
proves that $\text{FCP}(S,C)$ holds. Since this argument holds for
every set $S$ of size $t =\frac{3\delta-2}{2\delta-1}\alpha$, we
conclude that $\C$ has $\text{FCP}(t,C)$.
\section{Proof of lemma \ref{lem:matching}} \label{sec:proofs of lemma matching}

Consider a vector $\omega$ in the fundamental cone $\K = \K(H)$ of
the parity check matrix $H$. Without loss of generality, we may
assume that $S = \{1,2,\cdots,t\}$. For each $1\leq i\leq t$, let
the neighbors of the variable node $v_i$ in the $(p,q)$-matching on
$S$ be denoted by $c_1^i,c_2^i,\cdots,c_p^i$. The check nodes
$c_j^i$ are $p\times t$ distinct nodes. From the definition of $\K$,
if $\omega\in \K$, then for each $c_j^i$ we may write:

\begin{equation}
\omega_{i} \leq \sum_{l \in N(c_i^j)\setminus v_i} \omega_l, ~~
\forall 1\leq i \leq t~1\leq j\leq p. \label{eq: all Cone ineqs}
\end{equation}

\noindent We add all inequalities of (\ref{eq: all Cone ineqs}) for
  $1\leq i \leq t$ and  $1\leq j\leq p $. For $i \leq t$,
$\omega_i$ appears exactly $p$ times on the left hand side of the
sum and, at most $d_v - p$ times on the right. For $i > t$,
$\omega_i$ appears in at most $d_v-q$ inequalities and on the right
hand side. This comes directly from the definition of a
$(p,q)$-matching on the set $S$. Therefore
\begin{equation}
 p \sum_{i\in S} \omega_i \leq (d_v-p)\sum_{i\in S}\omega_i +
 (d_v-q)\sum_{i\in S^c}\omega_i,
\end{equation}
\noindent and thus,
\begin{equation}
 \frac{2p-d_v}{d_v-q} \sum_{i\in S} \omega_i \leq  \sum_{i\in S^c}\omega_i,
\end{equation}
\noindent which proves that $\C$ has the desired fundamental cone
property.

\section{Proof of Theorem \ref{thm:ADS robust 3}} \label{sec:proofs of thms}
%, \ref{thm:ADS
%robust 2} and \ref{thm:ADS robust 3} \end{center}
%
%\label{proofs of thms}

We denote the set of variable nodes and check nodes by $X_v$ and
$X_c$ respectively. For a fixed $\w\in[0,1]^T$, let $\mathcal{B}$ be
the set of all minimal $T$-local deviations, and $\mathcal{B}_i$ be
the set of minimal $T$-local deviations that result from a skinny
tree rooted at the variable node $v_i$. Also, assume $S$ is the
random set of flipped bits, when the flip probability is $p$.
Interchangeably, we also use $S$ to refer to the set of variable
nodes corresponding to the flipped bits indices. We are interested
in the probability that for all $\beta^{(\w)}\in\mathcal{B}$,
$f_C^{(S)}(\beta^{(\w)})\geq 0$. Recall that

\beq \nonumber f_C^{(S)}(\x) := \sum_{i\in S^c} x_i -C\sum_{i\in S}
x_i. \eeq

For simplicity we denote this event by $f_C^{(S)}(\mathcal{B})\geq
0$. Since the bits are flipped independently and with the same
probability, we have the following union bound

\beq\label{eq:union bound} \Prob\left(f_C^{(S)}(\mathcal{B})\geq
0\right) \geq 1-n\Prob\left(f_C^{(S)}(\mathcal{B}_1)\geq 0\right).
\eeq

\noindent Now consider the full tree of height 2T, that is rooted at
the node $v_1$, and contains every node $u$  in $\G$ that is no more
than $2T$ distant from $v$, i.e. $d(v_1,u)\leq 2T$. We denote this
tree by $B(v_1,2T)$. To every variable node $u$ of  $B(v_1,2T)$, we
assign a label, $I(u)$, which is equal to $-C\omega_{h(u)}$ if $u\in
S$, and is $\omega_{h(u)}$ if $u\in S^c$, where
$(\omega_0,\omega_{2},\cdots,\omega_{2T-2})=\w$. We can now see that
the event $f_C^{(S)}(\mathcal{B}_1)\geq 0$ is equivalent to the
event that for all skinny subtrees $\mathcal{T}$ of $B(v_1,2T)$ of
height $2T$, the sum of the labels on the variable nodes of
$\mathcal{T}$ is positive. In other words, if $\Gamma_1$ is the set
of all skinny trees of height $2T$ that are rooted at $v_1$, then
$f_C^{(S)}(\mathcal{B}_1)\geq 0$ is equivalent to:

\beq \min_{\mathcal{T}\in\Gamma_1}\sum_{v\in \mathcal{T}\cap X_v}
I(v) \geq 0. \eeq

We assign to each node $u$ (either check or variable node) of
$B(v_1,2T)$ a random variable $Z_u$, which is equal to the
contribution to the quantity
$\min_{\mathcal{T}\in\Gamma_1}\sum_{v\in \mathcal{T}\cap X_v} I(v)$
by the offspring of the node $u$ in the tree $B(v_1,2T)$, and the
node $u$ itself. The value of $Z_u$ for can be determined
recursively from all of its children. Furthermore, the distribution
of $Z_u$ only depends on the height of $u$ in $B(v_1,2T)$.
Therefore, to find the distribution of $Z_u$, we use
$X_0,X_1,\cdots,X_{T-1}$ as random variables with the same
distribution as $Z_u$ when $u$ is a variable node ($X_0$ is assigned
to the lowest level variable node) and likewise
$Y_1,Y_2,\cdots,Y_{T-1}$ for the check nodes. It then follows that:

\bea \nonumber Y_0 &=& \omega_0\eta, \\
     \nonumber X_i &=& \min\{Y_i^{(1)},\dots,Y_i^{(d_c-1)}\}~~\forall i>0, \\
     \nonumber Y_i &=&
     \omega_i\eta+X_{i-1}^{(1)}+\cdots+X_{i-1}^{(d_v-1)}~~\forall i>0, \\
\eea \noindent where $X^{(j)}$s are independent copies of a random
variable $X$, and $\eta$ is a random variable that takes the value
$-C$ with probability $p$ and value $1$ with probability $1-p$. It
follows that

\bea\nonumber \Prob\left(f_C^{(S)}(\mathcal{B}_1)\leq 0\right) &=
&\Prob\left(
X^{(1)}_{T-1}+\cdots+X^{(d_v)}_{T-1} \leq 0 \right)\\
&\leq& (\stexp(e^{-tX_{T-1}}))^{d_v}.
\label{P(f_C^)}
\eea

\noindent The last inequality is by Markov inequality and is true
for all $t>0$. The rest of the proof we bring here is basically
appropriate modifications of the derivations of \cite{ADS} for the
Laplace transform evolution of the variables $X_i$s and $Y_i$s, to
account for a non-unitary  robustness factor $C$.  By upper bounding
the Laplace transform of the variables recursively it is possible to
show that (see Lemma 8 of \cite{ADS}, the argument is completely the
same for our case)

\begin{align}\label{eq:lemma8 of ADS} \nonumber \stexp{e^{-tX_i}}\leq
&   \left(\stexp{e^{-tX_j}}\right)^{(d_v-1)^{i-j}} \\ &\prod_{0\leq
k\leq
 i-j-1}\left((d_c-1)\stexp{e^{-t\omega_{i-k}\eta}}\right)^{(d_v-1)^k},\end{align}

\noindent for all $1\leq j\leq i< T$.

If we take the weight vector as $\omega =
(1,2,\cdots,2^j,\rho,\rho,\cdots,\rho)$ for some integer $1 \leq
j<T$, and use equation (\ref{eq:lemma8 of ADS}), we obtain:

\begin{eqnarray*}
&\stexp{e^{-tX_{T-1}}}  \leq &   (\stexp{e^{-tX_j}})^{(d_v-1)^{T-j-1}} \\
%&\quad & (\stexp{e^{-tX_j}})^{(d_v-1)^{T-j-1}}\left((d_c-1)\stexp{e^{-t\rho\eta}}\right)^{\frac{(dv-1)^{T-j-1}-1}{d_v-2}},
&\quad & \cdot \left((d_c-1)\stexp{e^{-t\rho\eta}}\right)^{\frac{(dv-1)^{T-j-1}-1}{d_v-2}}.
\end{eqnarray*}
\noindent $\rho$ and $t$ can be chosen to jointly minimize
$\stexp{e^{-tX_j}}$ and $\stexp{e^{-t\rho\eta}}$ in the above, which along with (\ref{P(f_C^)}) results in
\bea \nonumber \Prob(f_C^{\mathcal{S}}\left(\mathcal{B}_1\leq 0\right)) &\leq&
(\stexp{e^{-tX_{T-1}}})^{d_v}\\ \nonumber &\leq&
\gamma^{-d_v/(d_v-2)}\times c^{d_v (d_v-1)^{T-j-1}}, \eea
\noindent where $\gamma =
(d_c-1)\frac{C+1}{C}(1-p)(\frac{C.p}{1-p})^{1/(C+1)}$ and $c =\gamma^{1/(d_v-2)}\min_{t\geq
0}{\stexp{e^{-tX_j}}}$. If $c <1$, then probability of error tends to zero as stated in Theorem \ref{thm:ADS robust 3}.

\end{document}